\documentclass[11pt]{article}

\usepackage[margin=0.9in]{geometry}
\usepackage{amsmath,amssymb,amsthm,mathtools}
\usepackage{booktabs}
\usepackage{microtype}
\usepackage[hidelinks]{hyperref}
\usepackage{xurl}
\newtheorem{theorem}{Theorem}
\newtheorem{lemma}[theorem]{Lemma}

\theoremstyle{remark}

\newcommand{\TEL}{\operatorname{TEL}}
\newcommand{\OPT}{\operatorname{OPT}}

\newcommand{\Rdet}{R_{\mathrm{det}}}

\hypersetup{
  pdftitle={A Tight Linear Deterministic Competitive Ratio for Fully Online KV-Cache Scheduling},
  pdfauthor={Ian D'Ambrosio}
}

\title{A Tight Linear Deterministic Competitive Ratio for\\
Fully Online KV-Cache Scheduling}
\author{Ian D'Ambrosio\\
\small Nth Research Collective\\
\small \texttt{ian@nthresearch.org}}
\date{11 August 2026}

\begin{document}
\maketitle

\begin{abstract}
Jaillet et al. introduced a fully online model for batching nonpreemptive LLM
requests under a growing KV-cache memory constraint.  For total end-to-end
latency they proved that every deterministic algorithm has competitive ratio
$\Omega(\sqrt n)$, while the elementary sequential upper bound is $n$.  We
close this gap.  Let $\Rdet(n,M)$ be the optimal deterministic ratio for
exactly $n$ requests at memory $M$, and let
$\Rdet(n)=\sup_M\Rdet(n,M)$.  For every $n\geq2$ we prove
\[
  \frac{n-1}{12}\leq\Rdet(n)\leq n,
\]
so $\Rdet(n)=\Theta(n)$.  The lower bound releases one memory-filling long
request, observes its deterministic start time, and then releases $n-1$ wide
one-token requests halfway through the long run.  No short request can overlap
the long one, whereas a hindsight schedule runs the two groups in the opposite
order when useful.  The hard instance uses the explicit fixed memory
$M=2(n-1)n$.  The upper bound is achieved by a uniform causal serial policy.
The exact model, causality argument, both comparator branches, and quantifier
order are machine-checked in Lean 4.  Exact finite controls and replay commands
accompany the proof.
\end{abstract}


\section{Introduction}

Continuous batching in large-language-model inference permits many requests to
generate tokens in the same GPU iteration.  The principal coupling is memory:
the KV cache of a request grows by one unit after each generated token and is
released only when that request finishes.  A scheduler must therefore decide
which arrived prompts to admit while accounting for the future memory growth
of every admitted request.

Jaillet et al.~\cite{JailletEtAl2026} formalized this problem for one GPU and
total end-to-end latency.  Request sizes are known at arrival, future arrivals
are adversarial, and an admitted request must advance in every subsequent
round until completion.  They proved an $\Omega(\sqrt n)$ deterministic lower
bound.  No matching general upper bound was known.  The trivial policy that
serves requests serially gives an $n$ upper bound, leaving a polynomial gap.

We prove that the elementary upper order is tight.  The lower-bound mechanism
uses requests with one output token but prompt size approximately half the GPU
memory.  Such requests cannot overlap each other.  Releasing them after a
deterministic scheduler commits to a memory-filling long request creates a
single-machine bottleneck inside the batching model.  The hindsight optimum
avoids the bottleneck by choosing the better of long-first and shorts-first
orders.

The exact quantifier order matters.  For fixed $M$, the online scheduler knows
$M$ but not the eventual number of requests.  We first fix a single hard memory
$M=2(n-1)n$ and only then quantify over every deterministic causal scheduler at
that memory.  Thus the lower bound applies directly to
$\sup_M\inf_A$, without exchanging a supremum and an infimum.

\paragraph{Contribution and scope.}
The result determines the worst-over-memory asymptotic ratio.  It does not
determine the tight two-parameter function $\Rdet(n,M)$ for every fixed $M$.
The hard family also uses prompts of size $M/2$, so restricted-prompt regimes
remain separate questions.

\section{Model and competitive ratio}\label{sec:model}

Fix an integer memory capacity $M\geq1$.  A finite instance $I$ consists of
labeled requests
\[
  q=(a_q,s_q,o_q),
\]
where $a_q\in\mathbb Z_{\geq0}$ is the arrival time and
$s_q,o_q\in\mathbb Z_{>0}$ are the prompt and output lengths.  Individual
feasibility requires $s_q+o_q\leq M$.  At time $a_q$, the complete triple is
revealed.  Before that time the scheduler has no information about $q$.

Time is discrete.  If request $q$ starts in round $b_q$, it is active in the
rounds
\[
  b_q,b_q+1,\ldots,b_q+o_q-1
\]
and completes at $C_q=b_q+o_q$.  In active round $t$, its memory usage is
\begin{equation}\label{eq:memory}
  s_q+(t-b_q)+1.
\end{equation}
Thus its active-round memory values are
$s_q+1,\ldots,s_q+o_q$.  A schedule is feasible if $b_q\geq a_q$ for every
request and
\begin{equation}\label{eq:cap}
  \sum_{q:\,b_q\leq t<b_q+o_q}
    \bigl(s_q+(t-b_q)+1\bigr)\leq M
  \qquad\text{for every }t.
\end{equation}
This start-of-round convention is the one used by the executable checker.  It
is the one-step reindexing of the post-round convention in the integer program
of~\cite{JailletEtAl2026}; both have the same memory states, peak, completion
time, and latency.

Once started, a request cannot pause, restart, migrate, or be rejected.  Every
active request advances by one token in the unique batch each round.  The cost
of a schedule is
\[
  \TEL(I)=\sum_{q\in I}(C_q-a_q).
\]
The hindsight optimum $\OPT(I)$ knows the complete instance but obeys the same
release, nonpreemption, progress, and memory rules.

A deterministic scheduler at memory $M$ assigns starts to every finite valid
instance.  It is \emph{causal} if, whenever two instances have exactly the same
requests revealed through time $t$, every revealed request has the same
started-by-$t$ decision in the two instances.  This observational condition
allows decisions at time $t$ to use all requests arriving at $t$, but no future
request.  The scheduler is one rule across every finite request count and does
not know the final $n$.

For a scheduler $A$, define its exact-$n$, fixed-memory factor by
\[
  \operatorname{CR}_A(n,M)
  =\sup_{I:\,|I|=n}
    \frac{\TEL(I;A)}{\OPT(I)},
\]
where the supremum ranges over valid instances at memory $M$.  A scheduler
that fails to complete a valid finite instance has infinite factor.  Set
\[
  \Rdet(n,M)=\inf_A\operatorname{CR}_A(n,M),
  \qquad
  \Rdet(n)=\sup_{M\geq1}\Rdet(n,M).
\]

\begin{theorem}\label{thm:main}
For every integer $n\geq2$,
\[
  \frac{n-1}{12}\leq\Rdet(n)\leq n.
\]
Consequently, $\Rdet(n)=\Theta(n)$.
\end{theorem}

\section{The uniform serial upper bound}\label{sec:upper}

We first prove the upper bound using a deliberately conservative serial policy.
Give requests distinct labels and order them lexicographically by
$(a_q,\operatorname{id}_q)$.  Define
\begin{equation}\label{eq:serial}
  b_q=a_q+\sum_{p\prec q}o_p.
\end{equation}
This is not asserted to be work-conserving.  It is useful because it is uniform,
causal, and its cost is immediate.

\begin{lemma}\label{lem:serial}
The starts in \eqref{eq:serial} form a feasible causal schedule.  For every
valid $n$-request instance,
\[
  \TEL(I;A_{\mathrm{ser}})\leq
  n\sum_{q\in I}o_q.
\]
\end{lemma}

\begin{proof}
If $p\prec q$, then $a_p\leq a_q$, and the predecessor sum for $q$ contains
the predecessor sum for $p$ together with $o_p$.  Hence
\[
  b_p+o_p\leq b_q.
\]
No two requests overlap.  Each request is individually feasible, so the whole
schedule obeys \eqref{eq:cap}.

Every predecessor of $q$ arrives no later than $q$.  Therefore $b_q$ depends
only on requests revealed by time $a_q$, proving causality even though future
requests and the final count are unknown.  Finally,
\[
  C_q-a_q=\sum_{p\prec q}o_p+o_q
  \leq\sum_{p\in I}o_p.
\]
Summing over the $n$ requests proves the claim.
\end{proof}

Every feasible schedule has $C_q-a_q\geq o_q$ for each request.  Thus
$\OPT(I)\geq\sum_qo_q$, and Lemma~\ref{lem:serial} yields
\begin{equation}\label{eq:upper}
  \Rdet(n,M)\leq n
  \qquad\text{for every }M.
\end{equation}

\section{The wide-short lower bound}\label{sec:lower}

Fix $k\geq1$, set $n=k+1$, and fix the even memory
\begin{equation}\label{eq:M}
  M=2k(k+1).
\end{equation}
Let $A$ be an arbitrary deterministic causal scheduler at this already-fixed
memory.

Release one long request
\[
  L=(0,1,M-1).
\]
If $A$ never starts this individually feasible singleton, its competitive
factor is infinite.  Otherwise let $b$ be its start time.  At
\begin{equation}\label{eq:release}
  r=b+M/2,
\end{equation}
release $k$ short requests
\[
  S_j=(r,M/2,1),\qquad 0\leq j<k.
\]

\begin{lemma}[Causal commitment]\label{lem:causal}
The long request starts at time $b$ in the extended instance.
\end{lemma}

\begin{proof}
Through time $b$, the singleton and extended instances reveal exactly the same
request.  Causality implies that the long request has started by $b$ in the
extended instance.  If its extended-instance start were some $t<b$, the two
instances would also have identical revealed prefixes through $t$, forcing the
singleton long request to have started by $t$.  This contradicts its singleton
start $b$.
\end{proof}

\begin{lemma}[Forced wait]\label{lem:wait}
Every short request starts no earlier than $b+M-1$ and has latency at least
$M/2$.  Consequently,
\begin{equation}\label{eq:alg-lower}
  \TEL(I;A)\geq kM/2.
\end{equation}
\end{lemma}

\begin{proof}
Suppose a short request starts at time $t<b+M-1$.  Release feasibility gives
$t\geq r$.  The still-active long request then uses
\[
  1+(t-b)+1\geq M/2+2
\]
memory, while the short request uses $M/2+1$.  Their combined usage is at least
$M+3$, contradicting \eqref{eq:cap}.  Hence a short starts at or after
$b+M-1$, completes one round later, and has latency at least
\[
  (b+M-1)+1-r=M/2.
\]
Summing over the $k$ shorts proves \eqref{eq:alg-lower}.
\end{proof}

It remains to upper-bound the hindsight optimum by an explicit schedule.  The
shorts have output one and peak memory $M/2+1$, so scheduling them one per round
is feasible.

\paragraph{Late release.}
If $r\geq M-1$, start the long request at time zero and start short $j$ at
$r+j$.  The long finishes before any short begins.  This schedule has
\[
  O_1=M-1+\sum_{j=0}^{k-1}(j+1)
  \leq M+k^2\leq2M.
\]

\paragraph{Early release.}
If $r<M-1$, start short $j$ at $r+j$ and start the long request at $r+k$.
The shorts finish before the long begins.  Its cost is at most
\[
  O_2\leq r+k+(M-1)+k^2\leq4M,
\]
using $r<M$, $k\leq M$, and $k^2\leq M$, all immediate from
\eqref{eq:M}.  Thus in both cases a feasible hindsight schedule has cost at
most $6M$.  Combining this deliberately coarse common bound with
\eqref{eq:alg-lower} gives
\[
  \frac{\TEL(I;A)}{\OPT(I)}
  \geq\frac{kM/2}{6M}=\frac{k}{12}.
\]
Since the memory in \eqref{eq:M} was fixed before $A$ was chosen,
\[
  \Rdet(k+1,M)\geq k/12.
\]
Taking the supremum over memory proves the lower bound in
Theorem~\ref{thm:main}.  Equation~\eqref{eq:upper} proves the upper bound.

\section{Formal and executable verification}\label{sec:verification}

The formal development uses Lean 4.32.2 and Mathlib 4.32.2.  Its canonical
human-quantifier declaration is
\begin{center}
\small\ttfamily
OnlineKvCache.\allowbreak
deterministic\_worst\_memory\_ratio\_theta\_linear\_human\_quantifiers.
\end{center}
It fixes $M=2k(k+1)$ before quantifying over every lower-bound scheduler at that
memory.  The upper theorem quantifies over every fixed $M$.  The supporting
files formalize valid instances, exact per-round memory, TEL, revealed-prefix
causality, both comparator schedules, and positive-denominator cancellation.
Lean reports only the standard axioms \texttt{propext},
\texttt{Classical.choice}, and \texttt{Quot.sound} for the audited declaration.

From the repository root, replay with:
\begin{verbatim}
cd formal/lean_project
~/.elan/bin/lake build \
  LeanProject.OnlineKvCacheHumanQuantifiers
cd ../..
python3 -m unittest \
  tests.test_online_kv_cache_scheduling_eval \
  tests.test_bootstrap_online_kv_cache_sources \
  tests.test_online_kv_cache_linear_formal \
  tests.test_online_kv_cache_human_quantifiers_formal \
  tests.test_formal
python3 -m evaluators.online_kv_cache_scheduling_eval \
  --verify output/tight-online-kv-cache-scheduling/exact-controls.json
\end{verbatim}
The replay passes 30 focused tests and returns
\texttt{TIGHT\_ONLINE\_KV\_CACHE\_EXACT\_CONTROLS\_PASS}.
The canonical hashes are:
\begin{center}
\footnotesize
human-quantifier proof\\
\nolinkurl{d5b7a90989339d5a0ba0536ee1c49a9b91c41f9c22914b6096f78fa1c4180391}\\[2pt]
exact controls\\
\nolinkurl{ca8fd42aa4364ed0faee2c77f41703d84e86b38952ebb41bc6999adae2609a12}.
\end{center}
The exact finite checker independently replays schedules and computes offline
optima on bounded instances.  Its negative controls reject a corrupted memory
progress convention, distinguish TEL from total completion time, and expose
the effect of allowing kill and restart.  A fresh independent review rebuilt
the formal modules, reproduced the controls, checked the source-to-Lean time
shift and quantifier order, and repeated the primary prior-art search.

\section{Related work and scope}\label{sec:related}

The closest theorem is Jaillet et al.'s $\Omega(\sqrt n)$ deterministic lower
bound in the same adversarial-arrival, known-output, nonpreemptive TEL model
\cite{JailletEtAl2026}.  Wang, Ye, and Zhou study the same growing-memory
geometry with all jobs available initially~\cite{WangYeZhou2026}.  Kong et al.
analyze geometry-aware policies in burst and Poisson regimes rather than the
exact adversarial ratio~\cite{KongEtAl2026}.

Feng, Yang, and Zhang permit killing and restarting unfinished jobs in a
non-clairvoyant all-at-once model~\cite{FengYangZhang2026}.  Their subsequent
regime-aware routing result allows arbitrary online arrivals but minimizes
total completion time $\sum_q C_q$, again with kill and restart
\cite{FengLiuYangZhang2026}.  Although
\[
  \sum_q C_q=\TEL+\sum_q a_q,
\]
a multiplicative guarantee for total completion only implies
\[
  \TEL(A)\leq c\TEL(\OPT)+(c-1)\sum_q a_q.
\]
The additive arrival term can be arbitrarily large relative to optimal flow
time, so that theorem does not transfer to the present target.  Stochastic,
fluid, throughput, and preemptive packing models provide useful comparisons but
change either the benchmark or the permitted actions
\cite{AoEtAl2026,ImKulkarniMunagala2015,MitzenmacherShahout2025}.

The present result leaves several questions open.  Most importantly, it does
not determine the tight fixed-memory dependence of $\Rdet(n,M)$.  The lower
family proves linear hardness at $M=2(n-1)n$ and uses prompts of size $M/2$.
Sharper bounds for small fixed memory, bounded prompt sizes, or bounded size
spread may have different asymptotic behavior.  The constants $1/12$ and $1$
in Theorem~\ref{thm:main} are not optimized.

\paragraph{Data and code availability.}
The Lean sources, exact evaluator, source-pin bootstrap, controls, tests, and
paper source form the accompanying reproducibility package.  No empirical data
or private service is required to check the theorem.

\paragraph{Research-system disclosure.}
Beyond, the research system operated by Nth Research Collective, materially
assisted with literature retrieval, hypothesis generation, exact computation,
formalization, proof critique, adversarial review, and drafting.  Ian
D'Ambrosio selected the target, directed the investigation, reconciled the
evidence, verified the final claims, and accepts full responsibility for the
manuscript.  Beyond is not an author.


\begingroup
\small
\begin{thebibliography}{99}

\bibitem{AoEtAl2026}
 R.~Ao, G.~Luo, D.~Simchi-Levi, and X.~Wang,
\newblock Optimizing LLM inference: Fluid-guided online scheduling with memory
constraints,
\newblock arXiv:2504.11320v4, 2026.
\newblock \url{https://arxiv.org/abs/2504.11320v4}.

\bibitem{FengYangZhang2026}
Y.~Feng, Z.~Yang, and Y.~Zhang,
\newblock Competitive non-clairvoyant KV-cache scheduling for LLM inference,
\newblock arXiv:2601.22996v1, 2026.
\newblock \url{https://arxiv.org/abs/2601.22996v1}.

\bibitem{FengLiuYangZhang2026}
Y.~Feng, S.~Liu, Z.~Yang, and Y.~Zhang,
\newblock General non-clairvoyant KV-cache scheduling via regime-aware routing,
\newblock arXiv:2607.09248v1, 2026.
\newblock \url{https://arxiv.org/abs/2607.09248v1}.

\bibitem{ImKulkarniMunagala2015}
S.~Im, J.~Kulkarni, and K.~Munagala,
\newblock Competitive flow time algorithms for polyhedral scheduling,
\newblock in \emph{Proceedings of the 56th Annual IEEE Symposium on Foundations
of Computer Science}, 2015, pp.~506--524.
\newblock \url{https://doi.org/10.1109/FOCS.2015.38}.

\bibitem{JailletEtAl2026}
P.~Jaillet, J.~Jiang, K.~Mellou, M.~Molinaro, C.~Podimata, and Z.~Zhou,
\newblock Online scheduling for LLM inference with KV cache constraints,
\newblock arXiv:2502.07115v5, 2026.
\newblock \url{https://arxiv.org/abs/2502.07115v5}.

\bibitem{KongEtAl2026}
L.~Kong, Q.~Qi, Y.~Ye, and Z.~Zhou,
\newblock Geometry-aware online scheduling for LLM serving: From theoretical
bound to system practice,
\newblock arXiv:2606.22327v2, 2026.
\newblock \url{https://arxiv.org/abs/2606.22327v2}.

\bibitem{MitzenmacherShahout2025}
M.~Mitzenmacher and R.~Shahout,
\newblock Queueing, predictions, and LLMs: Challenges and open problems,
\newblock arXiv:2503.07545v1, 2025.
\newblock \url{https://arxiv.org/abs/2503.07545v1}.

\bibitem{WangYeZhou2026}
M.~Wang, Y.~Ye, and Z.~Zhou,
\newblock LLM serving optimization with variable prefill and decode lengths,
\newblock arXiv:2508.06133v4, 2026.
\newblock \url{https://arxiv.org/abs/2508.06133v4}.

\end{thebibliography}
\endgroup
\end{document}